\documentclass[11pt]{article}
\usepackage[leqno]{amsmath}
\usepackage{fullpage}

\usepackage{amssymb}
\usepackage{amsthm}
\usepackage{enumerate}
\usepackage{verbatim}
\usepackage{mathrsfs}
\usepackage{dsfont}
\usepackage{graphicx}
\usepackage{soul}

\newcommand{\ds}{\displaystyle}
\newcommand{\ts}{\textstyle}

\newcommand{\mr}{\mathrm}

\newcommand{\mc}{\mathcal}
\newcommand{\mb}{\mathbf}

\newcommand{\uhr}{\upharpoonright}

\newcommand{\tsfrac}[2]{{\ts\frac{#1}{#2}}}

\let\Ex\undefined
\let\Pr\undefined

\DeclareMathOperator*{\Ex}{\mathds{E}}
\DeclareMathOperator*{\Pr}{\mathds{P}}

\newcommand{\Supp}{\mathrm{Supp}}
\newcommand{\defiff}{\stackrel{\text{\tiny{def}}}{\iff}}

\usepackage{mathtools}

\newcommand{\fieldfont}[1]{\mathbb{#1}}
\newcommand{\N}{\fieldfont{N}}

\newtheoremstyle{theorem-style}
  {}
  {}
  {\slshape}
  {}
  {\bf}
  {.}
  {.5em}
  {}

\newtheorem{thm}{Theorem}

\newtheorem{la}[thm]{Lemma}

\newtheorem{main-la}[thm]{Main Lemma}

\newtheorem{cor}[thm]{Corollary}

\newtheorem{fact}[thm]{Fact}

\theoremstyle{definition}
\newtheorem{df}[thm]{Definition}

\renewcommand{\|}{|}

\newcommand{\PARITY}{\textsc{parity}}

\newcommand{\NCone}{\mathsf{NC^1}}
\newcommand{\ACone}{\mathsf{AC^1}}
\newcommand{\Ppoly}{\mathsf{P/poly}}

\newcommand{\as}{\mathsf{as}}

\begin{document}

\title{The Average Sensitivity of Bounded-Depth Formulas}
\author{Benjamin Rossman\footnote{National Institute of Informatics (Tokyo, Japan) and Simons Institute (Berkeley, CA). \texttt{rossman@nii.ac.jp}}}
\date{\today}
\maketitle{}

\begin{abstract}
We show that unbounded fan-in boolean formulas of depth $d+1$ and size $s$ have average sensitivity $O(\frac{1}{d}\log s)^d$.  In particular, this gives a tight $2^{\Omega(d(n^{1/d}-1))}$ lower bound on the size of depth $d+1$ formulas computing the \textsc{parity} function.  These results strengthen the corresponding $2^{\Omega(n^{1/d})}$ and $O(\log s)^d$ bounds for circuits due to H{\aa}stad (1986) and Boppana (1997).  Our proof technique studies a random process where the Switching Lemma is applied to formulas in an efficient manner.
\end{abstract}

\section{Introduction}

We consider boolean circuits with unbounded fan-in AND and OR gates and negations on inputs.  {\em Formulas} are the class of tree-like circuits in which all gates have fan-out 1.  {\em Size} of circuits (including formulas) is measured by the total number of gates.  {\em Depth} is the maximum number of gates on an input-to-output path.

Lower bounds against bounded-depth circuits were first proved in the 1980s \cite{Ajtai83,Furst84,Yao85,hastad1986almost}, culminating in a tight size-depth tradeoff for circuits computing the $\PARITY$ function. The technique, based on random restrictions, applies more generally to boolean functions with high average sensitivity.

\begin{thm}[H{\aa}stad \cite{hastad1986almost}]\label{thm:hastad}
Depth $d+1$ circuits computing $\PARITY$ have size $2^{\Omega(n^{1/d})}$.
\end{thm}

\begin{thm}[Boppana \cite{Boppana97}]\label{thm:boppana}
Depth $d+1$ circuits of size $s$ have average sensitivity $O(\log s)^d$.
\end{thm}

In this paper, we prove stronger versions of these results for bounded-depth formulas:

\begin{thm}\label{thm:main1}
Depth $d+1$ formulas computing $\PARITY$ have size $2^{\Omega(d(n^{1/d}-1))}$.
\end{thm}

\begin{thm}\label{thm:main2}
Depth $d+1$ formulas of size $s$ have average sensitivity $O(\frac{1}{d}\log s)^d$.
\end{thm}

Theorems \ref{thm:main1} and \ref{thm:main2} directly strengthen Theorems \ref{thm:hastad} and \ref{thm:boppana} in light of the following

\begin{fact}\label{fact:sep}
Every depth $d+1$ circuit of size $s$ is equivalent to a depth $d+1$ formula of size at most $s^d$.
\end{fact}

Theorems \ref{thm:hastad}, \ref{thm:boppana}, \ref{thm:main1}, \ref{thm:main2} are asymptotically tight, since $\PARITY$ is computable by depth $d+1$ circuits (resp.\ formulas) of size $n2^{O(n^{1/d})}$ (resp.\ $2^{O(d(n^{1/d}-1))}$).  

\newpage

The main tool in the proof of Theorems \ref{thm:hastad} and \ref{thm:boppana} is H{\aa}stad's Switching Lemma \cite{hastad1986almost}.  The Switching Lemma states that every small-width CNF or DNF simplifies, with high probability under a random restriction, to a small-depth decision tree.  This yields lower bounds against bounded-depth {\em circuits} via a straightforward depth-reduction argument. In this paper we show how the Switching Lemma can be applied more efficiently to bounded-depth {\em formulas}, though in a less straightforward manner.  

In more detail: for independent uniformly distributed random $\sigma \in \{0,1\}^n$ (``assignment'') and $\tau \in [0,1]^n$ (``timestamp''), we consider the family of restrictions $\{R^{\sigma,\tau}_p\}_{0 \le p \le 1}$ (i.e.\ functions $[n] \to \{0,1,\ast\}$ representing partial assignments to input variables $x_1,\dots,x_n$) where $R^{\sigma,\tau}_p$ sets the variable $x_i$ to $\sigma_i$ if $\tau_i < p$ and leaves $x_i$ unset if $\tau_i \ge p$.  In the usual application of the Switching Lemma to circuits of depth $d+1$, all subcircuits of depth $k+1$ are hit with the restriction $R^{\sigma,\tau}_{p_k}$ for a fixed sequence $p_1 > \dots > p_d$ (typically $p_k = n^{-k/(d+1)}$).  In this paper we achieve sharper bounds against formulas by hitting each subformula $\Phi$ with the restriction $R^{\sigma,\tau}_{\mb q(\Phi)}$ where the parameter $\mb q(\Phi)$ ($=\mb q^{\sigma,\tau}(\Phi)$) is defined inductively, according to a random process indexed by subformulas of $\Phi$.  Our technical main theorem is a tail bound on $\mb q(\Phi)$, viewed as a random variable determined by $\sigma$ and $\tau$.

After preliminary definitions in \S\ref{sec:prelims}, we state and prove our technical main theorem in \S\ref{sec:main} and \S\ref{sec:tail}. As a corollaries, we derive Theorem \ref{thm:main1} in \S\ref{sec:parity} and Theorem \ref{thm:main2} in \S\ref{sec:as}. In \S\ref{sec:vs} we state a further corollary of our results on the relative power of formulas vs.\ circuits.

\section{Preliminaries}\label{sec:prelims}

$\N = \{0,1,2,\dots\}$. $[n] = \{1,\dots,n\}$. $\exp(\lambda) = \mr e^\lambda$. 

\subsection{Formulas}
A {\em formula} is a finite rooted tree whose leafs (``inputs'') are labeled by literals (i.e.\ variables $x_i$ or negated variables $\neg x_i$) and whose non-leafs (``gates'') are labeled by AND or OR. (Gates have unbounded fan-in.) Every formula $\Phi$ computes a boolean function on the same set of variables.

The {\em size} of a formula $\Phi$, denoted by $\|\Phi\|$, is the number of gates in $\Phi$. (Note that every lower bound on {\em size} is also a lower bound on {\em leafsize}, i.e.,\ the number of leaves in a formula.)   The {\em depth} of $\Phi$ is the maximum number of gates on an input-to-output path. Formulas of depth $0$ are literals; formulas of depth $1$ are clauses (i.e.\ an AND or OR of literals). We are often interested in formulas of depth $\ge 2$ and speak of ``depth $d+1$'' where $d$ is an arbitrary positive integer.

\subsection{Boolean functions and restrictions}

A {\em restriction} is a function $\varrho : [n] \to \{0,1,\ast\}$, viewed as a partial assignment of boolean input variables $x_1,\dots,x_n$ to $0$, $1$ or $\ast$ (meaning ``unset'').  For a boolean function $f : \{0,1\}^n \to \{0,1\}$, the restricted function $f{\uhr}\varrho : \{0,1\}^{\varrho^{-1}(\ast)} \to \{0,1\}$ is defined in the usual way.  For $p \in [0,1]$, we write $\mc R_p$ for the distribution on restrictions $\varrho$ where $\Pr[\ \varrho(i)=\ast\ ]=p$ and $\Pr[\ \varrho(i)=0\ ]=\Pr[\ \varrho(i)=1\ ]=(1-p)/2$ independently for all $i \in [n]$.

\subsection{Average sensitivity and decision-tree depth}

The {\em average sensitivity} $\as(f)$ of a boolean function $f$ is the expected number of input bits that, when flipped, change the output of $f$, starting with a random input assignment.  

The {\em decision-tree depth} $\mathsf{D}(f)$ of $f$ is the minimum depth of a decision tree which computes $f$; in particular, $\mathsf{D}(f) = 0$ iff $f$ is constant. 
Two elementary facts which we will use later (see \cite{Boppana97}): for every boolean function $f$,
\begin{align}\label{eq:fact1}
  &&&&\as(f) &\le \mathsf{D}(f) &&\text{(i.e.\ average sensitivity is at most decision-tree depth}),&&&&\\
  \label{eq:fact2}
  &&&&\Ex_{\varrho \sim \mc R_p}[\ \as(f{\uhr}\varrho)\ ] 
  &= p{\cdot}\as(f) &&\text{for all } 0 \le p \le 1.
\end{align}

H{\aa}stad's Switching Lemma relates random restrictions and decision-tree depth. We give a somewhat nonstandard statement (the usual statement is in terms of width-$k$ CNFs and width-$\ell$ DNFs).

\begin{la}[Switching Lemma \cite{hastad1986almost}]\label{la:hastad}
Let $k,\ell \in \N$.
Suppose $f$ is the $\mr{AND}$ or $\mr{OR}$ of an arbitrary family $\{f_i\}$ of boolean functions with $\mathsf D(f_i) \le k$ for all $i$.
Then for all $0 \le p \le \frac12$, 
\[
  \Pr_{\varrho \sim \mc R_p}
  [\
    \mathsf D(f{\uhr}\varrho) \ge \ell
  \ ]
  \le (5pk)^\ell.
\]
\end{la}

\section{A random process associated with formulas}\label{sec:main} 

\begin{df}
Let $\sigma \in \{0,1\}^n$ (``assignment'') and $\tau \in [0,1]^n$ (``timestamp'') be independent uniformly distributed random variables.  
For $0 \le p \le 1$, let $R^{\sigma,\tau}_p : [n] \to \{0,1,\ast\}$ be the restriction
\[
  R^{\sigma,\tau}_p(i) := 
  \begin{cases}
    \sigma_i  &\text{if } \tau_i > p,\\
    \ast      &\text{if } \tau_i \le p.
  \end{cases}
\]
\end{df}

We regard the family of restrictions $\{R^{\sigma,\tau}_p\}_{0 \le p \le 1}$ as a stochastic process where the parameter $p$ represents a ``time'' which starts at $1$ and decreases to $0$.  At the initial time $p=1$, the assignment $\sigma$ is fully masked (i.e.\ $R^{\sigma,\tau}_1$ is all $\ast$'s).  As $p$ decreases, the values of $\sigma$ are gradually unmasked, until the final time $p=0$ when $\sigma$ is fully revealed (i.e.\ $R^{\sigma,\tau}_0 = \sigma$).  Of course, for any fixed $p$, $R^{\sigma,\tau}_p$ is simply a random restriction with distribution $\mc R_p$.  

\begin{df}[Main Definition]
For all formulas $\Phi$, we define the ``stopping time'' $\mb q^{\sigma,\tau}(\Phi) \in [0,1]$ by the following induction:
\begin{itemize}
  \item
    If $\Phi$ has depth $0$ (i.e.\ $\Phi$ is a variable or negated variable), then $\mb q^{\sigma,\tau}(\Phi) := 1$.
  \item
    If $\Phi$ is $\mr{AND}(\Psi_1,\dots,\Psi_m)$ or $\mr{OR}(\Psi_1,\dots,\Psi_m)$, then 
    \[
      \mb q^{\sigma,\tau}(\Phi) :=
      \ds
      \frac{\mb p^{\sigma,\tau}(\Phi)}
      {14{\cdot}
      \mb k^{\sigma,\tau}(\Phi)}
    \]
    where \qquad
    $\ds  \mb p^{\sigma,\tau}(\Phi) 
      := \vphantom{\big|}
      \min_i \mb q^{\sigma,\tau}(\Psi_i),
      \qquad
      \mb k^{\sigma,\tau}(\Phi) := \vphantom{\Big|}
      \max\{1,\,
      \max_i \mathsf D(\Psi_i{\uhr}R^{\sigma,\tau}_{\mb p^{\sigma,\tau}(\Phi)})
      \}.
    $
\end{itemize}
\end{df}

For the sake of readability, we will suppress $\sigma$ and $\tau$ whenever possible and simply write $\mb q(\Phi)$, $\mb p(\Phi)$, $\mb k(\Phi)$.  However, the reader should keep in mind that these random variables are determined, for all formulas $\Phi$, by a single pair of $\sigma$ of $\tau$.  (We will continue to write $\sigma$ and $\tau$ when referring to restrictions $R^{\sigma,\tau}_p$.)

We view $\mb q(\Phi)$ as the stopping time for a stochastic process indexed by formulas $\Phi$.  For $\Phi$ of depth $0$, $\mb q(\Phi)$ is the initial time $1$ (when all variables are masked).  For $\Phi$ of depth $\ge 1$, $\mb q(\Phi)$ is defined in terms of two auxiliary parameters:
\begin{itemize}
\item
$\mb p(\Phi)$ is the 
most advanced (i.e.\ minimum) stopping time $\mb q(\Psi)$ among children $\Psi$ of $\Phi$.
\item
$\mb k(\Phi)$ is the maximum decision-tree depth among children $\Psi$ of $\Phi$ upon being hit with the restriction $R^{\sigma,\tau}_{\mb p(\Phi)}$. (For technical reasons, we set $\mb k(\Phi) = 1$ in the event that $D(\Psi{\uhr}R^{\sigma,\tau}_{\mb p(\Phi)}) = 0$ for all $\Psi$.) 
\end{itemize}

If $\Phi$ is an AND (resp.\ OR), then $\Phi{\uhr}R^{\sigma,\tau}_{\mb p(\Phi)}$ is a $\mb k(\Phi)$-CNF (resp.\ DNF).  The choice of definition $\mb q(\Phi) = \mb p(\Phi)/14{\cdot}\mb k(\Phi)$ allows us to apply the Switching Lemma to $\Phi{\uhr}R^{\sigma,\tau}_{\mb p(\Phi)}$.  This is made precise by the following lemma. (Since the dependence on $\sigma$ and $\tau$ is crucial here, we use explicit notation: $\mb q^{\sigma,\tau}(\Phi)$, etc.)

\begin{la}\label{la:switching}
Let $\Phi$ be a formula of depth $\ge 1$ and let $q \in \Supp(\mb q^{\sigma,\tau}(\Phi))$ (i.e.\ $q = \mb q^{\sigma,\tau}(\Phi)$ for some $\sigma \in \{0,1\}^n$ and $\tau \in [0,1]^n$).  Then for all $0 \le \alpha \le 1$ and $\ell \in \N$,
\[
  \Pr_{\sigma,\tau}
  \Big[\ \mathsf D(\Phi{\uhr}R^{\sigma,\tau}_{\alpha q}) \ge \ell
  \ \Big|\ \mb q^{\sigma,\tau}(\Phi) = q
  \ \Big]
  \le 
  \bigg(\frac{\alpha}{\mr e}\bigg)^\ell.
\]
\end{la}

\begin{proof}
Fix $\Phi$ and $q$ as in the hypothesis of the lemma. Since $\Phi$ has depth $\ge 1$, it is the AND or OR of formulas $\Psi_i$. Let
\[
  I := 
  \left\{\,
  (p,\varrho,k) : 
  \begin{aligned}
  &q = p/14k\text{ and there exist } 
  \sigma \in \{0,1\}^n 
  \text{ and } \tau \in [0,1]^n\\
  &\text{such that }
  \mb p^{\sigma,\tau}(\Phi)=p,\
  R^{\sigma,\tau}_{p}=\varrho \text{ and }
  \mb k^{\sigma,\tau}(\Phi)=k
  \end{aligned}
  \,\right\}.
\]
Note that $I$ is nonempty and indexes a partition of the event $\{\mb q^{\sigma,\tau}(\Phi) = q\}$ into subevents $\{\mb p^{\sigma,\tau}(\Phi)=p,\
  R^{\sigma,\tau}_{p}=\varrho
  \text{ and }
  \mb k^{\sigma,\tau}(\Phi)=k\}$. 

To prove the lemma, consider any $(p,\varrho,k) \in I$.  Conditioning on this subevent, we can view $R^{\sigma,\tau}_{\alpha q}$ as the composition of 
$\varrho$ and an independent random restriction $\theta \sim \mc R_{\alpha/14k}$.  Since $\Phi{\uhr}\varrho$ is an AND or OR of functions $\Psi_i{\uhr}\varrho$ of decision-tree depth $\le k$, Lemma \ref{la:hastad} implies
\begin{align*}
  \Pr_{\sigma,\tau}
  \Big[\ \mathsf D(\Phi{\uhr}R^{\sigma,\tau}_{\alpha q}) \ge \ell
  \ \Big|\ 
  \mb p^{\sigma,\tau}&(\Phi)=p,\
  R^{\sigma,\tau}_{p}=\varrho
  \text{ and }
  \mb k^{\sigma,\tau}(\Phi)=k
  \ \Big]\\
  &=
  \Pr_{\theta \sim \mc R_{\alpha/14k}}
  \Big[\ \mathsf D((\Phi{\uhr}\varrho){\uhr}\theta) \ge \ell\ \Big]
  \le
  \bigg(5\bigg(\frac{\alpha}{14k}\bigg)k\bigg)^\ell
  \le 
  \bigg(\frac{\alpha}{\mr e}\bigg)^\ell.\qedhere
\end{align*}
\end{proof}

\section{Tail bound on $\mb q(\Phi)$}\label{sec:tail}

Our technical main theorem is a tail bound on the random variable $\mb q(\Phi)$ ($= \mb q^{\sigma,\tau}(\Phi)$) where the randomness is over independent uniform $\sigma \in \{0,1\}^n$ and $\tau \in [0,1]^n$.  We state the result first with asymptotic notation.

\begin{thm}\label{thm:main}
For every depth $d+1$ formula $\Phi$ and $0 < \lambda \le 1$,
\[
  \Pr
  \big[\
    \mb q(\Phi) \le \lambda
  \ \big]
  \le
  \frac{\|\Phi\|}{\exp(\Omega(d\lambda^{-1/d}) - O(d))}.
\]
\end{thm}

\newpage

In order to have a useable induction hypothesis, we restate Theorem \ref{thm:main} with explicit constants:

\addtocounter{thm}{-1}

\begin{thm}[more precisely]
For every depth $d+1$ formula $\Phi$ and $\ell > 0$,
\[
  \Pr
  \bigg[\
    \mb q(\Phi) \le \frac{1}{14^{d+1}\ell}
  \ \bigg]
  \le
  \|\Phi\|\frac{C^d}{\exp(\mr e^{-2}d\ell^{1/d})}
\]
where
$\ds
  C =
  1 + \sum_{i=0}^\infty 
  \Bigg( 
  \frac{1}{\exp(\mr e^{i-1}-(i+1)\mr e^{-2})} + 
  \sum_{j=0}^\infty \frac{1}{\exp((j+1)\mr e^{i-1} - (i+j+2)\mr e^{-2})} 
  \Bigg)
  \approx
  7.83.
$
\end{thm}

\begin{proof}
We first note that the theorem is trivial if $\ell < \mr e^d$ (as the RHS is $> (C/\exp(\mr e^{-1}))^d > 1$ since $C > \exp(\mr e^{-1})$).  Therefore, we assume that $\ell \ge \mr e^d$.  We argue by induction on $d$.  

Consider the base case $d=1$ where $\Phi$ is a depth $2$ formula. Note that $\mb q(\Psi) = 1/14$ for each depth $1$ subformula $\Psi$ of $\Phi$; hence $\mb p(\Phi) = 1/14$.  Also, each $\Psi$ is the AND or OR of decision-trees of depth $1$; so by Lemma \ref{la:hastad},
\[
  \Pr_{\sigma,\tau}
  \Big[
  \
    \mathsf D(\Psi{\uhr}R^{\sigma,\tau}_{1/14}) \ge \ell
  \
  \Big]
  =
  \Pr_{\varrho \sim \mc R_{1/14}}
  \Big[
  \
    \mathsf D(\Psi{\uhr}\varrho) \ge \ell
  \
  \Big]
  \le 
  \bigg(\frac{1}{\mr e}\bigg)^\ell.
\]
Since $\mb q(\Phi) = \mb p(\Phi)/14{\cdot}\mb k(\Phi) = 1/14^2{\cdot}\mb k(\Phi)$, we have
\begin{align*}
  \Pr\Big[\
    \mb q(\Phi) \le \frac{1}{14^2\ell}
  \ \Big]
  =
  \Pr\Big[\
    \mb k(\Phi) \ge \ell
  \ \Big]
  &=
  \Pr\Big[\
  \bigvee_{\Psi}
    \mathsf D(\Psi{\uhr}R^{\sigma,\tau}_{\mb p(\Phi)}) \ge \ell
  \ \Big]\\
  &\le
  \sum_\Psi
  \Pr\Big[\
    \mathsf D(\Psi{\uhr}R^{\sigma,\tau}_{1/14}) \ge \ell
  \ \Big]\\
  &\le
  |\Phi|
  \frac{1}{\exp(\ell)}
  <
  |\Phi|\frac{C^d}{\exp(\mr e^{-2}d\ell^{1/d})}.
\end{align*}

For the induction step, let $d \ge 2$ and assume the theorem holds for $d-1$. Let $\Phi$ be a formula of depth $d+1$.  Let $\Psi$ range over depth-$d$ subformulas of $\Phi$.  In particular, we have $\|\Phi\| = 1+\sum_{\Psi} \|\Psi\|$.

We will define a family of events denoted $\mc A$ and $\mc B_i$ ($i \in \N$) and $\mc C_{i,j}$ ($i,j \in \N$) and show that the union of these events covers the event $\{\mb q(\Phi) \le \frac{1}{14^{d+1}\ell}\}$. We will then bound the probability of each of these events and show that the (infinite) sum of these probabilities is at most $\|\Phi\| \smash{\frac{C^d}{\exp(\mr e^{-2}d\ell^{1/d})}}$. 

For all $i \in \N$, define $k_i$ and $\alpha_i$ by
\begin{align*}
  k_i &:= \mr e^{i-1}\ell^{1/d},
  \qquad\quad
  \alpha_i := \frac{k_i}{14^d\ell} 
  \ 
  \bigg({=}\ \frac{1}{14^d\mr e^{1-i}\ell^{(d-1)/d}}\bigg).
\end{align*}
Events $\mc A$ and $\mc B_i$ and $\mc C_{i,j}$ ($i,j \in \N$) are defined as follows:
\begin{align*}
  \mc A 
  &\ \defiff\
  \Big(\mb p(\Phi) \le \alpha_0\Big),
  \vphantom{\bigvee_\Psi}\\
  \mc B_i
  &\ \defiff\
  \bigvee_\Psi
  \Big(\mb q(\Psi) \le \alpha_{i+1}\Big)
  \wedge
  \Big(\mathsf D(\Psi{\uhr}R^{\sigma,\tau}_{\mb q(\Psi)}) \ge k_i\Big),\\
  \mc C_{i,j}
  &\ \defiff\
  \bigvee_\Psi
  \Big(\alpha_{i+j+1} < \mb q(\Psi) \le \alpha_{i+j+2}\Big)
  \wedge
  \Big(\mathsf D(\Psi{\uhr}R^{\sigma,\tau}_{\alpha_{i+1}}) \ge k_i\Big).
\end{align*}

\noindent{\underline{Claim}}:\quad If
$\ds
  \mb q(\Phi) \le \frac{1}{14^{d+1}\ell}$, then
$\ds \mc A \vee \bigvee_{i = 0}^\infty 
  \bigg(\mc B_i \vee \bigvee_{j = 0}^\infty \mc C_{i,j}\bigg).
$\bigskip

\noindent{\underline{Proof of claim}}:\quad 
Assume $\mb q(\Phi) \le 1/14^{d+1}\ell$ and further assume that $\mc A$ does not hold. Clearly there exists a unique $i \in \N$ such that $\alpha_i < \mb p(\Phi) \le \alpha_{i+1}$ (since $\alpha_i$ is eventually $>1$). Since $\mb q(\Phi) = \mb p(\Phi)/14{\cdot}\mb k(\Phi)$, we have $\mb k(\Phi) > \alpha_i 14^d \ell = k_i$.  Note that $k_i \ge k_0 = \mr e^{-1} \ell^{1/d} \ge 1$ (using the assumption that $\ell \ge \mr e^d$).  Since $\mb k(\Phi) = \max\{1,\max_\Psi \mathsf D(\Psi{\uhr}R^{\sigma,\tau}_{\mb p(\Phi)})\}$, it follows that there exists a $\Psi$ such that $\mathsf D(\Psi{\uhr}R^{\sigma,\tau}_{\mb p(\Phi)}) \ge k_i$. 

Fix an arbitrary choice of $\Psi$ such that $\mathsf D(\Psi{\uhr}R_{\mb p(\Phi)}) \ge k_i$.  There are two cases to consider: either $\mb q(\Psi) \le \alpha_{i+1}$ or $\alpha_{i+j+1} < \mb q(\Psi) \le \alpha_{i+j+2}$ for some $j \in \N$.
\begin{itemize}
\item
Assume $\mb q(\Psi) \le \alpha_{i+1}$.  In this case, we have $\mathsf D(\Psi{\uhr}R_{\mb p(\Phi)}) \le \mathsf D(\Psi{\uhr}R_{\mb q(\Psi)})$ since $\mb p(\Phi) \le \mb q(\Psi)$.  Therefore, $\mathsf D(\Psi{\uhr}R^{\sigma,\tau}_{\mb q(\Psi)}) \ge k_i$.  We conclude that $\mc B_i$ holds.  
\item
Assume $\alpha_{i+j+1} < \mb q(\Psi) \le \alpha_{i+j+2}$ for some $j \in \N$.  We have $\mathsf D(\Psi{\uhr}R_{\mb p(\Phi)}) \le \mathsf D(\Psi{\uhr}R_{\alpha_{i+1}})$ since $\mb p(\Phi) \le \alpha_{i+1}$.  Therefore, $\mathsf D(\Psi{\uhr}R^{\sigma,\tau}_{\alpha_{i+1}}) \ge k_i$.  We conclude that $\mc C_{i,j}$ holds.
\end{itemize}
This concludes the proof of the claim.\bigskip

To complete the proof of the theorem, we will bound the probabilities of events $\mc A$, $\mc B_i$ and $\mc C_{i,j}$ and take a union bound.  We ignore the fact that all but finitely many of these events have zero probability, since $\Pr[\ \mc B_i\ ] = 0$ (resp.\ $\Pr[\ \mc C_{i,j}\ ] = 0$) for all $\alpha_i > 1$ (resp.\ $\alpha_{i+j+1} > 1$).  Instead, we show that $\Pr[\ \mc B_i\ ]$ is exponentially decreasing in $i$, while $Pr[\ \mc C_{i,j}\ ]$ is exponentially decreasing in $j$ and doubly exponentially decreasing in $i$.

We first bound the probability of $\mc A$:
\begin{align*}
  \Pr[\,\mc A\,] 
  =
  \Pr\Big[\ \bigvee_{\Psi}
   \mb q(\Psi) 
  \le \frac{1}{14^d\mr e\ell^{(d-1)/d}}\ \Big]
  &\le
  \sum_{\Psi}
  \Pr\Big[\ \mb q(\Psi) \le \frac{1}{14^d\mr e\ell^{(d-1)/d}}\ \Big]
  \\
  &\le
  \|\Phi\|
  \frac{C^{d-1}}{\exp(\mr e^{-2}(d-1)\mr e^{1/(d-1)}\ell^{1/d})}
  &&\text{(induction hypothesis)}\\
  &\le 
  \|\Phi\|
  \frac{C^{d-1}}{\exp(\mr e^{-2}d\ell^{1/d})}
  &&\text{(using $\mr e^{1/(d-1)} \ge \tsfrac{d}{d-1}$).}
\end{align*}

We next bound the probability of $\mc B_i$:
\begin{align*}
  \Pr[\, \mc B_i\, ]
  &=
  \Pr\Big[\
  \bigvee_\Psi
  \Big(\mb q(\Psi) \le \alpha_{i+1}\Big)
  \wedge
  \Big(\mathsf D(\Psi{\uhr}R^{\sigma,\tau}_{\mb q(\Psi)}) \ge k_i\Big)
  \ \Big]
  \\
  &\le
  \sum_\Psi
  \Pr\Big[\
  \mb q(\Psi) \le \alpha_{i+1}
  \ \Big]
  \Pr\Big[\
  \mathsf D(\Psi{\uhr}R^{\sigma,\tau}_{\mb q(\Psi)}) \ge k_i
  \ \Big|\
  \mb q(\Psi) \le \alpha_{i+1}
  \ \Big]
  \hspace{-1.2in}\\
  &\le
  \bigg(\frac{1}{\mr e}\bigg)^{k_i}
  \sum_\Psi
  \Pr\Big[\
  \mb q(\Psi) \le \alpha_{i+1}
  \ \Big]
  &&\text{(Lemma \ref{la:switching})}\\
  &=
  \frac{1}{\exp(\mr e^{i-1}\ell^{1/d})}
  \sum_\Psi
  \Pr\Big[\
  \mb q(\Psi) \le \frac{1}{14^d\mr e^{-i}\ell^{(d-1)/d}}
  \ \Big]\\
  &\le
  \frac{1}{\exp(\mr e^{i-1}\ell^{1/d})}
  \|\Phi\|
  \frac{C^{d-1}}{\exp(\mr e^{-2}(d-1)\mr e^{-i/(d-1)}\ell^{1/d})}
  &&\text{(induction hypothesis)}\\
  &\le
  \frac{1}{\exp(\mr e^{i-1}\ell^{1/d})}
  \|\Phi\|
  \frac{C^{d-1}}{\exp(\mr e^{-2}(d-1)\ell^{1/d} - i\mr e^{-2}\ell^{1/d})}
  &&\text{($\mr e^{-i/(d-1)} \ge 1-\tsfrac{i}{d-1}$)}\\
  &= 
  \frac{1}{\exp((\mr e^{i-1}-(i+1)\mr e^{-2})\ell^{1/d})}
  \|\Phi\|
  \frac{C^{d-1}}{\exp(\mr e^{-2}d\ell^{1/d})}\\
  &\le
  \frac{1}{\exp(\mr e^{i-1}-(i+1)\mr e^{-2})}
  \|\Phi\|
  \frac{C^{d-1}}{\exp(\mr e^{-2}d\ell^{1/d})}.
\end{align*}
The last inequality uses the assumption $\ell^{1/d} \ge 1$ as well as the nonnegativity of $\mr e^{i-1} - (i+1)\mr e^{-2}$ for all $i \in \N$.

Finally, we bound the probability of $\mc C_{i,j}$:
\begin{align*}
  \Pr[\, \mc C_{i,j}\, ]
  &=
  \Pr\Big[\
  \bigvee_\Psi
  \Big(\alpha_{i+j+1} < \mb q(\Psi) \le \alpha_{i+j+2}\Big)
  \wedge
  \Big(\mathsf D(\Psi{\uhr}R^{\sigma,\tau}_{\alpha_{i+1}}) \ge k_i\Big)
  \Big]\hspace{-1in}\\
  &\le
  \sum_\Psi
  \Pr\Big[\
    \mb q(\Psi) \le \alpha_{i+j+2}
  \ \Big]
  \Pr\Big[\
  \mathsf D(\Psi{\uhr}R^{\sigma,\tau}_{\alpha_{i+1}}) \ge k_i
  \ \Big|\ 
    \alpha_{i+j+1} < \mb q(\Psi) \le \alpha_{i+j+2}
  \Big]
  \hspace{-1in}\\
  &\le
  \left(\frac{\alpha_{i+1}/\alpha_{i+j+1}}{\mr e}\right)^{k_i}
  \sum_\Psi
  \Pr\Big[\
    \mb q(\Psi) \le \alpha_{i+j+2}
  \ \Big]
  &&\text{(Lemma \ref{la:switching})}\\
  &=
  \frac{1}{\exp((j+1)\mr e^{i-1}\ell^{1/d})}
  \sum_\Psi
  \Pr\Big[\
    \mb q(\Psi) \le \frac{1}{14^d\mr e^{-(i+j+1)}\ell^{(d-1)/d}}
  \ \Big]\hspace{-1in}\\
  &\le
  \frac{1}{\exp((j+1)\mr e^{i-1}\ell^{1/d})}
  \|\Phi\|
  \frac{C^{d-1}}{\exp(\mr e^{-2}(d-1)\mr e^{-(i+j+1)/(d-1)}\ell^{1/d})}
  &&\text{(ind.\ hyp.)}\\
  &\le
  \frac{1}{\exp((j+1)\mr e^{i-1}\ell^{1/d})}
  \|\Phi\|
  \frac{C^{d-1}}{\exp(\mr e^{-2}(d-1)\ell^{1/d} - (i+j+1)\mr e^{-2}\ell^{1/d})}\hspace{-1in}\\
  &=
  \frac{1}{\exp(((j+1)\mr e^{i-1} - (i+j+2)\mr e^{-2})\ell^{1/d})}
  \|\Phi\|
  \frac{C^{d-1}}{\exp(\mr e^{-2}d\ell^{1/d})}\\
  &\le
  \frac{1}{\exp((j+1)\mr e^{i-1} - (i+j+2)\mr e^{-2})}
  \|\Phi\|
  \frac{C^{d-1}}{\exp(\mr e^{-2}d\ell^{1/d})}.
\end{align*}
The last inequality uses the assumption $\ell^{1/d} \ge 1$ and the nonnegativity of $(j+1)\mr e^{i-1} - (i+j+2)\mr e^{-2}$ for all $i,j \in \N$.

We finish the proof by taking a union bound:
\begin{align*}
  \Pr\Big[\ \mb q(\Phi) \le \frac{1}{14^{d+1}\ell}\ \Big]
  &\le 
  \Pr[\, \mc A\,] + 
  \sum_{i=0}^\infty
  \bigg(
  \Pr[\, \mc B_i\,] + \sum_{j=0}^\infty \Pr[\,\mc C_{i,j}\, ]\bigg)
  \le
  \|\Phi\|
  \frac{C^d}{\exp(\mr e^{-2}d\ell^{1/d})}.
  \qedhere
\end{align*}
\end{proof}

\section{PARITY}\label{sec:parity}

We use the results of the last section to prove our lower bound for the $\PARITY$ function.

\setcounter{thm}{2}

\begin{thm}[restated]\label{thm:formulas2}
Depth $d+1$ formulas computing $\PARITY$ require size $\exp(\Omega(d(n^{1/d}-1)))$.
\end{thm}

\begin{proof}
Suppose $\Phi$ is a depth $d+1$ formula computing $\PARITY$.  Then
\[
  \Pr_{\varrho \sim \mc R_{1/n}}
  \big[\
    \Phi{\uhr}\varrho \text{ is non-constant}
  \ \big]
  = 
  1 - \Big(1-\frac{1}{n}\Big)^n
  >
  1 - \frac{1}{\mr e}.
\]
On the other hand, by Theorem \ref{thm:main} and Lemma \ref{la:switching},
\begin{align*}
  \Pr_{\varrho \sim \mc R_{1/n}}
  \big[\
    \Phi{\uhr}\varrho \text{ is non-constant}
  \ \big]
  &=
  \Pr_{\sigma,\tau}
  \big[\
    \mathsf D(\Phi{\uhr}R^{\sigma,\tau}_{1/n}) \ge 1
  \ \big]\\
  &\le
  \Pr
  \big[\
    \mathsf D(\Phi{\uhr}R^{\sigma,\tau}_{\max\{1/n,\mb q(\Phi)\}}) \ge 1
  \ \big]
  \vphantom{\Big|}\\
  &\le
  \Pr
  \big[\
    \mb q(\Phi) \le 1/n
  \ \big]
  +
  \Pr
  \big[\
    \mathsf D(\Phi{\uhr}R^{\sigma,\tau}_{\mb q(\Phi)}) \ge 1
  \ \big]
  \vphantom{\Big|}\\
  &\le
  \frac{\|\Phi\|}{\exp(\Omega(dn^{1/d})-O(d))}
  +
  \frac{1}{\mr e}.
\end{align*}
Therefore,
\[
  \|\Phi\| \ge \Big(1-\frac{2}{\mr e}\Big)
  \exp\Big(\Omega(dn^{1/d})-O(d)\Big).
\]
It follows that there exist universal constants $c_0,c_1>0$ (determined by the constants in the $\Omega(\cdot)$ and $O(\cdot)$) such that $\|\Phi\| \ge \exp(c_0d(n^{1/d}-1))$ in the regime $d \le c_1\ln n$.

In the regime $d > c_1\ln n$, we have $d(n^{1/d}-1) = \Theta(\ln n)$, more precisely,
\[
  \ln n < d(n^{1/d}-1) < c_1(e^{c_1}-1)\ln n.
\]
Note that $d(n^{1/d}-1)$ is decreasing in $d$ and $\lim_{d \to \infty} d(n^{1/d}-1) = \ln n$.  Invoking Khrapchenko's $n^2$ leafsize lower bound \cite{Khrap71} (which implies a (gate)size lower bound of $n$), we get a tight lower bound of $\exp(\Omega(d(n^{1/d}-1)))$ which is valid for all $d$ and $n$.
\end{proof}

\section{Average Sensitivity}\label{sec:as}

\setcounter{thm}{3}

\begin{thm}[restated]
Depth $d+1$ formulas of size $s$ have average sensitivity $O(\frac{1}{d}\ln s)^d$.
\end{thm}

\begin{proof}
Let $\Phi$ be a formula of depth $d+1$ and size $s$ (recall that size is the number of gates).  Assume $\as(\Phi) \ge 1$, since otherwise the theorem is trivial.  We further assume that $\Phi$ has bottom fan-in $\le s$; otherwise it is easily shown that $\as(\Phi) = O(\as(\Phi'))$ where $\Phi'$ is obtained from $\Phi$ by replacing every bottom AND (resp.\ OR) gate with fan-in $>s$ with $0$ (resp.\ $1$).  In particular, $\Phi$ has leafsize $\le s^2$, so it depends on $\le s^2$ distinct variables.

Letting $p = 1/\as(\Phi)$ and using facts (\ref{eq:fact1}) and (\ref{eq:fact2}), we have
\begin{align*}
  1 
  &= 
  p{\cdot}\as(\Phi)
  =
  \Ex_{\varrho \sim \mc R_p}\big[\ \as(\Phi{\uhr}\varrho)\ \big]
  \le
  \Ex_{\sigma,\tau}\big[\ \mathsf D(\Phi{\uhr}R^{\sigma,\tau}_p)\ \big]
  =
  \sum_{k=1}^{s^2}
  \Pr_{\sigma,\tau}\big[\ \mathsf D(\Phi{\uhr}R^{\sigma,\tau}_p) \ge k\ \big].
\end{align*}
For all $k \in \N$, by Theorem \ref{thm:main} and Lemma \ref{la:switching},
\begin{align*}
  \Pr_{\sigma,\tau}\big[\ \mathsf D(\Phi{\uhr}R^{\sigma,\tau}_p) \ge k\ \big]
  &\le
  \Pr_{\sigma,\tau}\big[\ \mathsf D(\Phi{\uhr}R^{\sigma,\tau}_{\max\{p,\mb q(\Phi)\}}) \ge k\ \big]
  \\
  &\le
  \Pr\big[\ \mb q(\Phi) \le p\ \big]
  +
  \Pr\big[\ 
  \mathsf D(\Phi{\uhr}R^{\sigma,\tau}_{\mb q(\Phi)}) 
  \ge k\ \big]\\
  &\le
  \frac{s}{\exp(\Omega(d{\cdot}\as(\Phi)^{1/d})-O(d))}
  +
  \frac{1}{\mr e^k}.
\end{align*}
Combining these inequalities, we have
\begin{align*}
  \exp(\Omega(d{\cdot}\as(\Phi)^{1/d})-O(d))
  &\le \frac{s^3}{1 - \sum_{k=1}^\infty 
  \mr e^{-k}}
  = 
  \frac{1-\mr e^{-1}}{1-2\mr e^{-1}}
  s^3
  =
  O(s^3).
\end{align*}
It follows that
$\Omega(d{\cdot}\as(\Phi)^{1/d})
  \le 
  3\ln s + O(d)$ and therefore $\as(\Phi) = O(\frac{1}{d}\ln s)^d$.
\end{proof}

\section{Formulas vs.\ Circuits}\label{sec:vs}

\setcounter{thm}{10}

Our lower bound for $\PARITY$ (Theorem \ref{thm:main1}) implies a separation between the power of depth $d+1$ formulas vs.\ circuits.  We write $\{$poly-size depth $d+1$ circuits/formulas$\}$ for the non-uniform complexity class of languages computable by $n^{O(1)}$-size depth $d+1$ circuits/formulas where $d(n)$ is an arbitrary function of $n$.

\begin{cor}
For all $d(n) = o(\log n)$ with $\lim_{n \to \infty} d(n) = \infty$,
\begin{equation}\label{eq:sep1}
  \{\text{poly-size depth $d+1$ formulas}\} \ne \{\text{poly-size depth $d+1$ circuits}\}.
\end{equation}
Moreover, for all $d \le C\frac{\log n}{\log\log n}$ (for some universal constant $C > 0$),
\begin{equation}\label{eq:sep2}
  \{\text{poly-size depth $d+1$ circuits}\} \nsubseteq \{\text{$n^{o(d)}$-size depth $d+1$ formulas}\}.
\end{equation}
\end{cor}

Separation (\ref{eq:sep1}) may be regarded as the depth $d+1$ analogue of the conjectured separation $\{$poly-size formulas$\} \ne \{$poly-size circuits$\}$, also known as $\NCone \ne \Ppoly$.  By Spira's theorem \cite{Spira71}, every poly-size formula is equivalent to a poly-size formula of depth $O(\log n)$; thus, extending (\ref{eq:sep1}) from depth $o(\log n)$ to depth $O(\log n)$ would imply $\NCone \ne \Ppoly$ (in fact $\NCone \ne \ACone$). 

For the smaller range of $d \le c\frac{\log n}{\log\log n}$, we get the stronger separation (\ref{eq:sep2}).  In light of Fact \ref{fact:sep}, this is the strongest possible separation between formulas and circuits of the same depth.

We remark that until recently not even the weak separation (\ref{eq:sep1}) was known to hold for any super-constant $d \nleq O(1)$.  The first progress on this question was made in \cite{rossman2014formulas}, where (\ref{eq:sep2}) was shown to hold for all $d \le \log\log\log n$ via a lower bound for \textsc{distance-$\log\log n$ st-connectivity}. In fact, the lower bound of \cite{rossman2014formulas} implies a much stronger result: for all $d \le \log\log\log n$,
\begin{equation}\label{eq:sep3}
  \{\text{poly-size depth $d+1$ circuits}\} \nsubseteq \{\text{$n^{o(d)}$-size depth $\tsfrac{\log n}{(\log\log n)^3}$ formulas}\}.
\end{equation}
It remains an open problem to push separation (\ref{eq:sep3}) to greater depths.

\section*{Acknowledgements}
My thanks to Rahul Santhanam, Rocco Servedio and Li-Yang Tan for valuable discussions and to the anonymous referees of FOCS'15 for their helpful feedback. This work was carried out while the author was a research fellow at the Simons Institute.


\begin{thebibliography}{10}

\bibitem{Ajtai83}
Mikl\'os Ajtai.
\newblock {$\Sigma^1_1$} formulae on finite structures.
\newblock {\em Annals of Pure and Applied Logic}, 24:1--48, 1983.

\bibitem{Boppana97}
Ravi B.\ Boppana.
\newblock The average sensitivity of bounded-depth circuits.
\newblock {\em Information Processing Letters}, 63(5):257--261, 1997.

\bibitem{Furst84}
Merrick~L. Furst, James~B. Saxe, and Michael Sipser.
\newblock Parity, circuits, and the polynomial-time hierarchy.
\newblock {\em Mathematical Systems Theory}, 17:13--27, 1984.

\bibitem{hastad1986almost}
Johan H{\aa}stad.
\newblock Almost optimal lower bounds for small depth circuits.
\newblock In {\em 18th Annual
  ACM Symposium on Theory of Computing}, pages
  6--20, 1986.
  
\bibitem{Khrap71}
V.M.\ Khrapchenko.
\newblock Complexity of the realization of a linear function in the case of $\Pi$-circuits.
\newblock
Math.\ Notes Acad.\ Sciences, 9:21--23, 1971.
  
\bibitem{rossman2014formulas}
Benjamin Rossman.
\newblock Formulas vs. circuits for small distance connectivity.
\newblock In {\em 46th Annual ACM Symposium on Theory of Computing},
  pages 203--212, 2014.  
  
\bibitem{Spira71}
P.M.\ Spira.
\newblock
On time-hardware complexity tradeoffs for Boolean functions.
\newblock  
In {\em 4th Hawaii Symposium on System Sciences},
pages 525--527, 1971.

\bibitem{Yao85}
Andrew C.C.\ Yao.
\newblock
Separating the polynomial-time hierarchy by oracles.
\newblock In {\em 26th Annual IEEE Symposium on Foundations of Computer Science}, pages 1--10, 1985.

\end{thebibliography}
\end{document}